\documentclass[10pt]{article}
\usepackage{fullpage}
\usepackage{cite}
\usepackage{amsmath,amssymb,amsfonts, amsthm, appendix}
\newtheorem{theorem}{Theorem}[section]
\renewcommand\thetheorem{\arabic{section}.\arabic{theorem}}
\newtheorem{lemma}[theorem]{Lemma}

\newtheorem{definition}[theorem]{Definition} 
\newtheorem{corollary}[theorem]{Corollary} 
 
\newtheorem{assumption}[theorem]{Assumption}

\usepackage{dsfont}

\usepackage{graphicx}
\usepackage{textcomp}
\usepackage{xcolor}
\usepackage{graphicx, float, amsmath, amsfonts, mathtools, bm, xcolor}
\usepackage{algorithm, algpseudocode, url}
\usepackage[colorlinks=true, allcolors=blue]{hyperref}

\def\BibTeX{{\rm B\kern-.05em{\sc i\kern-.025em b}\kern-.08em
    T\kern-.1667em\lower.7ex\hbox{E}\kern-.125emX}}
    
\usepackage{tikz}
\usepackage{pgfplots}
 \pgfplotsset{compat=1.17}
\usetikzlibrary{pgfplots.groupplots}

\pgfplotstableread[col sep = comma]{figures/isit_figs/active_passive_gauss.txt}\activepassivegauss
\pgfplotstableread[col sep = comma]{figures/isit_figs/active_passive_gauss_nointerp.txt}\activepassivegaussnointerp
\pgfplotstableread[col sep = comma]{figures/isit_figs/active_passive_laplace_nointerp_mc100.txt}\activepassivelapnointerp

\pgfplotstableread[col sep = comma]{figures/isit_figs/active_passive_laplacian.txt}\activepassivelaplacian

\pgfplotstableread[col sep = comma]{figures/isit_figs/noisecomp_gauss.txt}\noisecompgauss
\pgfplotstableread[col sep = comma]{figures/isit_figs/noisecomp_laplacian.txt}\noisecomplaplacian

\pgfplotstableread[col sep = comma]{figures/isit_figs/varcomp_gauss.txt}\varcompgauss
\pgfplotstableread[col sep = comma]{figures/isit_figs/varcomp_laplacian.txt}\varcomplaplacian

\pgfplotstableread[col sep = comma]{figures/isit_figs/noisecomp_gauss_mc100.txt}\noisecompgaussnointerp
\pgfplotstableread[col sep = comma]{figures/varcomp_laplacian_nointerp.txt}\varcomplaplaciannointerp

\usepackage{bookmark}
    \newcommand{\R}{\mathbb{R}}
    \newcommand{\init}{\mathrm{init}}
    \newcommand{\Y}{\bm{Y}}
    \newcommand{\Yhat}{\hat{\Y}}
    
    \newcommand{\U}{\bm{U}}
    \newcommand{\V}{\bm{V}}
    \newcommand{\e}{\bm{e}}
    \newcommand{\ep}{\mathbb{E}}
    \newcommand{\X}{\bm{X}}
    \newcommand{\Sig}{\bm{\Sigma}}
    \newcommand{\E}{\bm{E}}
        \renewcommand{\u}{\bm u}
    \renewcommand{\v}{\bm v}
\usepackage{xspace}    
    \newcommand{\norm}[1]{\left\| {#1} \right\|}
    \def\polylog{\operatorname{polylog}}
\newcommand{\algoname}{\textsf{UBAD}}

\newcommand{\x}{\bm{x}}
\begin{document}

\title{Uncertainty-Based Non-Parametric Active Peak Detection}

\author{Praneeth Narayanamurthy and Urbashi Mitra \\
\texttt{\{praneeth, ubli\}@usc.edu} \\ 
University of Southern California
}
\date{}
\sloppy
\maketitle

\begin{abstract}
Active, non-parametric peak detection is considered.  As a use case, active source localization is examined and an uncertainty-based sampling scheme algorithm to effectively localize the peak from a few energy measurements is designed. It is shown that under very mild conditions, the source localization error with $m$ actively chosen energy measurements scales as $O(\log^2 m/m)$. Numerically, it is shown that in low-sample regimes, the proposed method enjoys superior performance on several types of data and outperforms the state-of-the-art passive source localization approaches {and in the low sample regime, can outperform greedy methods as well.}
\end{abstract}



\section{Introduction}

Herein, we treat the problem of maximizing an unimodal function that is imperfectly sampled.  An important instantiation of this problem is peak detection or peak finding which is prevalent in many signal processing, computer vision and machine learning applications. In particular, we are motivated by the problem of source localization in the absence of a parametric model for the source signal.  This problem is relevant to acoustics, \cite{sourceloc_review_sound, sourceloc_review_sound_dl}, wireless networks \cite{sourceloc_review_network}, and medical imaging \cite{sourceloc_review_medical} \textit{etc.}  

Classical source detection and localization approaches require environmental models and/or energy models \cite{sourceloc1, sourceloc2, sourceloc3}; in many applications (\textit{e.g.} underwater and gas-source localization) such models are challenging to come by.  Herein, we leverage prior work \cite{umf_sourceloc}   which specifically exploited the unimodality of localization signals enabling a non-parametric approach. While the challenge of model knowledge is obviated in \cite{umf_sourceloc}, randomized sampling still requires a measurable number of samples.  Herein, we focus on {\bf active} sampling with the goal of strongly reducing the sample complexity. Active single source target localization has been well-studied  \cite{targetrecsurvey, targetest1, distillled_sense, active_target1}. These works, however, impose restrictive signal assumptions:
 separability of energy fields, full energy measurements and so on. 


In this work, we propose Uncertainty-Based Active Peak Detection (\algoname{}) to solve a much more general form of the problem to achieve active source localization. We adopt the overall algorithmic framework of \cite{umf_sourceloc} based on unimodal matrix completion. We enhance \cite{umf_sourceloc} by exploiting methods from \textit{active matrix completion} and \textit{multi-armed bandits.}


A Bayesian formulation for matrix completion \cite{smg_active} derives maximum-entropy based sampling schemes given the posterior distributions. Error bounds for uniform random sampling under different priors are computed in \cite{bayesian_mc_1, bayesian_mc_2}.  In contrast, a non-Bayesian formulation is undertaken in \cite{active_psd_mc} using Nystr\"{o}m sampling for positive semi-definite matrices. However, none of these works are readily applicable to the current problem setting since we consider a non-parametric, and extremely generic problem formulation.  Low-rank matrix completion in a bandits setting  \cite{bandit_mc_1, bandit_mc_2}  also assumes that rewards are sampled from a prior distribution as in the Bayesian case.  A seemingly promising work for our framework is that of  \cite{stoch_lowrank_bandits} where an elimination method is employed to find the maximum element of a low rank matrix; unfortunately, it is assumed that the singular vectors 
are a convex combination of some underlying latent factors. We observe that\cite{oned_source_loc} shows that although greedy search space reduction techniques are highly efficient for peak detection in the absence of noise, such methods fail in the presence of noise. Thus, we will leverage the idea of {\em optimism in the face of uncertainty} from multi-armed bandits \cite{bandit_algo} for our active sampling scheme.

The contributions of this work are:
\begin{enumerate}
\item We propose \algoname{}, a novel uncertainty-based algorithm to solve the problem of non-parametric, active peak detection.  
\item We show that as long as the energy function is non-increasing, our proposed method has favourable error bounds that hold under high generality.
\item We validate our theoretical claims through extensive numerical results and demonstrate that the proposed method outperforms existing (non-adaptive, adaptive but greedy) source localization algorithms
\end{enumerate}

\subsection{Notation}
We use the shorthand notation, $[k] := \{1, 2, \cdots, k\}$. We use lower ($\bm{x}$) and upper ($\bm{M}$) case boldface letters to denote vectors and matrices respectively. We index the $i$-th entry of a vector as $\x_i$ and $i,j$-th entry of a matrix as $\bm{M}_{i,j}$ respectively. Given a matrix, $\bm{M} \in \R^{m \times n}$, and a set $\mathcal{S} \subseteq [m] \times [n]$,  $\bm{M}_{\mathcal{S}} \in \R^{m\times n}$ sets the entries outside $\mathcal{S}$ to zero. For vectors (matrices), we use $\|\cdot\|$ to denote the  2- (induced 2-) norm unless specified otherwise. Throughout the paper we use $c, C$ to denote (possibly different) constants in each use.

\section{Problem Setting, Algorithm, and Main Result}
We next define the problem setting,  describe the algorithm, and provide our main theoretical results. 
\newcommand{\s}{\bm{s}}
\subsection{Problem Setting}
In this work, we consider the setting of a single source located at $\s^\ast \in \R^2$ (unknown) inside a target area, $\mathcal{A}$ with area with width $L$. We discretize the $ L \times L$ target area\footnote{We assume a centered and symmetric area for the sake of notational simplicity. All our results apply to general spaces as well.}  $\mathcal{A} = [-L/2, L/2] \times [-L/2, L/2]$ into $n \times n$ equally spaced grid points. Thus, the center of the $(i,j)$-grid square, is given by $\x_{i,j} := \left( \frac{L(2i-1)}{2n}, \frac{L(2j-1)}{2n}\right)^\intercal$.  We assume that the true energy measurement at a point $\x_{i,j} \in \mathcal{A}$ is given by some unknown, non-negative and monotonically non-increasing function, $h:\R^2 \times \R^2 \to \R$. 
We also assume that there can be additional measurement noise. Thus, the energy measurement at the center of the $(i,j)$-th grid cell is given as  
\begin{align}\label{eq:problem}
    \Y_{i,j} = h(\x_{i,j}, \s^\ast)  + \bm{z}_{i,j},
\end{align}
where, $\s^\ast := \x_{i^\ast, j^\ast}$ is the true source/peak location, the function $h(\cdot, \cdot)$ is the energy field and $\bm{z}_{i,j} \overset{i.i.d}{\sim} \mathcal{N}(0, \sigma_n^2)$ models measurement noise. In matrix form, the signal model can be  expressed as 
\begin{align}
    \Y := \bm{H}(\s^\ast) + \bm{Z}.
\end{align}

We notice that owing to the uniform discretization of the target area and the monotonic assumption on the energy field, $h(\cdot, \cdot)$, the matrix $\bm{H}(\s^\ast)$ is {\em unimodal}. We use the following definition of matrix unimodality in this work.

\begin{definition}[Unimodality]
A matrix $\bm{M}$ is said to be unimodal with the mode at $(i^\ast, j^\ast)$ if $\bm M_{1,j} \leq \bm M_{2,j} \cdots \leq \bm M_{i^\ast,j} \geq \bm M_{i^\ast +1, j} \geq \cdots \geq \bm M_{n,j}$ for all $j \in [n]$ and  $\bm M_{i,1} \leq \bm M_{i,2} \cdots \leq \bm M_{i,j^\ast} \geq \bm M_{i, j^\ast +1} \cdots \geq \bm M_{i,n}$ for all $i \in [n]$.
\end{definition}

The goal of this work is to efficiently estimate $\s^\ast$ by actively querying noisy energy measurements, $\Y_{i,j}$ informed by prior measurements.


\subsection{UBAD Algorithm}

We propose an uncertainty-based active sampling approach,  Uncertainty-Based Active Peak Detection (\algoname{}). For the initialization/exploration phase, we use the method of Latin Squares \cite{latin_squares}. For the peak detection, we adapt the method of \cite{umf_sourceloc} that transforms the peak detection problem into a unimodality constrained matrix completion operation. We first explain the two sampling schemes followed by the peak detection step in detail. 

 


\textbf{Initial Sampling:} We draw the initial set of samples, $\Omega_{\mathrm{init}} \subset [n] \times [n]$ with $|\Omega_{\mathrm{\init}}| = n$ through\footnote{For the sake of simplicity, in this paper we consider only square matrices after discretization. Rectangular matrices can be handled by a simply {\em stacking} multiple smaller Latin Squares along the appropriate dimension.} Latin Squares.  Thus, by construction, we sample the energy measurement in each row and column exactly once. We show that one is equally likely to sample each row/column. Intuitively, (in the single peak setting) this method gives us {\em sufficiently good} information about the location of the peak. We quantify this  notion in Lemma \ref{lem:init_lemma} provided in Sec. \ref{sec:main_res}.



\textbf{Active Sampling:} For the active (sequential) sampling stage, we devise a uncertainty-based sampling scheme. Specifically, under standard MC assumptions (incoherence of the matrix, randomly selected samples, and a large enough number of observations) \cite[Theorem 2]{uq_matcomp} derives distributional guarantees on the output of {any} MC algorithm. More precisely, assume that $\Y \overset{\text{SVD}}= \bm{U} \bm{\Sigma}_y \bm{V}^\intercal$ be a rank-$r$ matrix. Then, as long as we observe $O(n r^5 \polylog (n))$ measurements, with probability at least $ 1 - n^{-3}$, the output of a MC algorithm satisfies  
\begin{align*}
\hat{\Y}_{i,j} \sim \mathcal{N}(\Y_{i,j},C \sqrt{r/n} (\|\U^{(i)}\|_2^2 +  \|\V^{(j)}\|_2^2),
\end{align*}
where $\|\bm{U}^{(i)}\|$ denotes the $i$-th row of $\bm{U}$. 

There are two key challenges to be addressed here: (a) we consider the setting of low-samples and so, we do not observe a sufficient number of samples; (b) we do not have access to the true incoherence values (that determine the variance of the estimates). The first point requires a more thorough investigation into derivation of uncertainty bounds which we will consider as part of future work whereas the second challenge is addressed using the estimate of the incoherence instead of the actual incoherence.

\textbf{Peak Detection:} The remaining part of the algorithm is similar to that of the passive source localization problem studied in \cite{umf_sourceloc}. More specifically, we complete the energy matrix followed by estimating the location of the peak as the largest entry of the completed matrix. The complete pseudocode is provided in Algorithm \ref{algo:active_ucb}.

    \begin{algorithm}[t!]
    \caption{Uncertainty based Active Peak Detection (\algoname{})}
    \label{algo:active_ucb}
    \begin{algorithmic}[1]
    \Require $\Y \in \R^{n \times n}$ (energy matrix), $m$ (\# sequential samples)
    \State \textbf{Init:} $\Omega \gets $\Call{Latin-Squares}{$n$}
    \State $[\tilde{\bm u}^0, \tilde{\sigma}^2, \tilde{\bm y}^0] \leftarrow \mathrm{SVD}_1(\Y_{\Omega})$
    \State $\hat{\bm u}^0 \leftarrow  \tilde{\sigma} \tilde{\bm u}^0$, $\hat{\bm v}^0 \leftarrow \tilde{\sigma} \tilde{\bm v}^0$
    \For{$t \in [m]$}
    
    \State $\Yhat^{t-1} \leftarrow \hat{\bm u}^{t-1} (\hat{\bm v}^{t-1})^\intercal$
    \State $\hat\s_t = (\hat{x}^{t-1}, \hat{y}^{t-1}) \leftarrow \arg \max\limits_{i,j} |\Yhat^{t-1}_{i,j}|$ 
    \State $(i_t, j_t) \leftarrow \arg \max\limits_{(i,j) \in \Omega^c} |\Yhat^{t-1}_{i,j}| + (\hat{\bm{u}}^{t-1}_i)^2 + (\hat{\bm{v}}^{t-1}_j)^2$ 
    \State $\Omega \leftarrow \Omega \cup (i_t, j_t)$ \Comment{Update index set}
    \State Query next sample, $\Y_{i_t, j_t}$, update $\Y_{\Omega}$
    \State Solve $\Yhat^t \leftarrow \arg\min\limits_{\bm{M} \in \R^{n \times n}} \|\bm{M} - \Y_{\Omega}\|_{\ast}$ 
    \State Update $\hat{\bm u}^t (\hat{\bm v}^t)^\intercal \overset{SVD}{\leftarrow} \Yhat^t$
    
    \EndFor
    \Ensure $\hat{\s}_m $
    \end{algorithmic}
    \end{algorithm}
    

    \renewcommand{\Pr}{\mathrm{Pr}}
    
    \subsection{Main Result and Proof Sketch} \label{sec:main_res}
    
    Before showing the main result, we provide a few preliminaries. Let the rank-$1$ SVD of the energy matrix be given by $\Y := \sigma_y^2 \u \v^\intercal$ with $\|\u\|_2 = \|\v\|_2 = 1$. Define the maximum value of $\Y$, $b := \max_{i,j} \Y_{i,j}$. We define the row (column) differences, as 
    $\Delta_{k|l}^u := \Y_{i^\ast,l} -  \Y_{k,l}$ and $\Delta_{l|k}^v := \Y_{k,j^\ast} -  \Y_{k,l}$.
    

    \begin{theorem}[Sequential Sampling]\label{thm:main_result}
    Assume that the measurements satisfy \eqref{eq:problem}. Then,  with probability at least $0.9$, the source localization error for the sequential sampling satisfies
        \begin{align*}
        \ep\left[\frac{1}{m} \sum_{t=1}^m \| \hat\s_{\tau} - \s^\ast\|_2^2 \right] \leq C \sum_{k,l = 1}^n \frac{\gamma_{k,l}^u}{(\Delta_{k|l}^u)^2} 
        \cdot
        \frac{\gamma_{k,l}^v}{(\Delta_{l|k}^v)^2}
        \cdot \|\bm{c}_{k,l}- \bm{c}^\ast\|^2 \frac{\log^2 m}{m}  
    \end{align*}
    where, $\gamma^u_{k,l} := \u_k + 2b \Delta_{k|l}^u$ and $\gamma^v_{k,l} := \v_l + 2b \Delta_{l|k}^v$; $\bm{c}_{k,l} := (k,l)^\intercal$ and $\bm{c}^\ast := (i^\ast, j^\ast)^\intercal$. 
    \end{theorem}
    
      \textbf{Proof Sketch:} There are three main proof components: initialization, demonstration of maximum entropy sampling, and bounding the error of the sequential sampling stage. First consider the initialization step: 
    
    
    
    \begin{lemma}[Latin Squares Initialization]\label{lem:init_lemma}
    Consider the $n \times n$ rank $1$ matrix $\Y$. Let $\Omega_\init$ be the output of Latin-Squares. Then, with probability at least $1 - n^{-10}$, the error bound for the matrix is
    \begin{align}
        \|\Y - \Y_{\Omega_\init}\|_2 \leq 1.01 (1-1/n) \sigma_y^2 \|\u\|_1 \|\v\|_2
    \end{align}
    \end{lemma}
    The proof relies on carefully decomposing the residual matrix, $\Y - \Y_{\Omega\init}$ into a sum of independent sub-Gaussian random matrices, followed by the application of the Matrix Bernstein inequality \cite{tropp}. To the best of our knowledge, there are no existing results that consider just $n$ observations that are drawn through Latin-Squares based sampling. The Latin-squares sampling, in particular, is challenging to deal with since the entries of $\Omega_{\init}$ are not independent.

    Essentially, Lemma \ref{lem:init_lemma} serves as a proxy for a matrix completion step after the initial sampling stage and allows us to apply the distributional guarantees from \cite{uq_matcomp}.  Specifically, \cite[Theorem 2]{uq_matcomp} tells us that as long we sample $O(n \polylog{n})$ entries of the rank-$1$ matrix $\Y$, uniformly at random, the output of any matrix completion algorithm satisfies
    \begin{align}\label{eq:uq_assu}
    \hat{\Y}_{i,j} \sim \mathcal{N}(\Y_{i,j}, \frac{1}{\sqrt{n}}(|\u_i|^2 +  |\v_j|_2^2)),
    \end{align}    
    with probability at least $1 - n^{-3}$. In this paper,  we continue with the following assumption instead \footnote{Although \cite{uq_matcomp} assumes an $\epsilon$-accurate recovery, following the matrix completion step, we conjecture that the upper bound of Lemma \ref{lem:init_lemma} suffices to invoke \cite{uq_matcomp}, albeit with larger residual terms and lower probability of success.} 
    
    \begin{assumption}[Uncertainty Based Sampling]\label{prop:uncertainty}
    Under the conditions of Theorem \ref{thm:main_result},  choosing the next sample, $(i_t, j_t)$ at the $t$-th iteration chooses the location with maximum entropy. Furthermore, \eqref{eq:uq_assu} holds with high probability even after replacing $\u$ (and $\v$) with $\hat\u^t$ (and $\hat\v^t$).
    \end{assumption}


We now consider the  correctness of the sequential sampling step. The proof relies on a careful application of the regret bounds for the multi-armed Bandit problem wherein each arm has an unknown (and possibly distinct) variance \cite{mab_unknown_var}. Observe that the matrix completion step (Line 9 of \algoname{}) induces mutual independence in the estimates, $\Yhat^t$ across time. This directly follows from the fact that matrix completion algorithms typically begin with a {\em random Gaussian initialization}, followed by a refinement procedure. Additionally, we observe from \eqref{eq:uq_assu} that for any given $t$, each row (and column) is nearly statistically independent\footnote{This stems from the variance expression plus the fact that we do not impose any stochastic assumption on $\Y$.} of all others. The last key insight is that selecting the $t$-th sequential sample, $(i_t, j_t)$ is equivalent to selecting {\em an optimal} row, $i_t$ and a column $j_t$, which is further equivalent to playing two instances of the multi-armed Bandit problem -- one to select a row and the other to select a column. To see that at any given $t$, the two Bandit problems are independent, notice that  The above reason justifies the application of the regret bound, \cite[Theorem 1]{mab_unknown_var}. The remainder of the proof follows from careful algebraic manipulation. We provide the proof in the Appendix.

    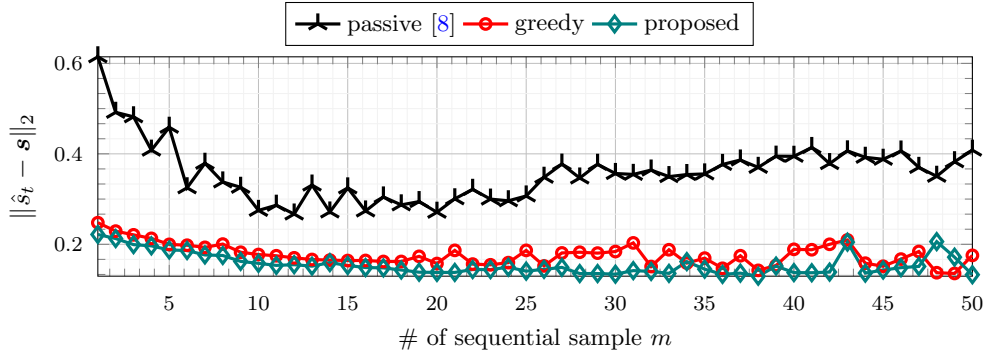
\begin{figure}[t!]
\centering
\begin{tikzpicture}
 \begin{groupplot}[
     group style={
        group size=1 by 1,
        x descriptions at=edge bottom,
        },
        enlargelimits=false,
        width = .8\linewidth,
        height=4.5cm,
        enlargelimits=false,
        grid=both,
    grid style={line width=.1pt, draw=gray!10},
    major grid style={line width=.2pt,draw=gray!50},
    minor x tick num=5,
    minor y tick num=5,
     ]
    
    \nextgroupplot[
    legend entries={
	passive \cite{umf_sourceloc},
	greedy,
    proposed,
	},
    legend style={at={(.75, 1.25)}},
    legend columns = 3,
    legend style={font=\small},
    xlabel={\small{\# of sequential sample $m$}},
    ylabel={{$\|\hat{s}_t - \bm{s}\|_2$}},
    title style={at={(0.5,-.35)},anchor=north,yshift=1},
    xticklabel style= {font=\footnotesize, yshift=-1ex},
    yticklabel style= {font=\footnotesize, xshift=-.5ex},
    ]
     \addplot [black, line width=1.2pt, mark=Mercedes star,mark size=4pt] table[x index = {0}, y index = {1}]{\activepassivegaussnointerp};
    \addplot [red, line width=1.2pt, mark=o,mark size=2pt] table[x index = {0}, y index = {2}]{\activepassivegaussnointerp};
    \addplot [teal, line width=1.2pt, mark=diamond,mark size=3pt] table[x index = {0}, y index = {3}]{\activepassivegaussnointerp};

\end{groupplot}
\end{tikzpicture}
\vspace*{-0.1in}
\caption{Source localization error with respect to the number of sequential samples observed.}
\label{fig:comp_active_passive}
\end{figure}

    This general result can be simplified for the special cases of Gaussian and Laplacian fields. 
    
     \begin{corollary}[Special Cases]\label{cor:special_case}
        \begin{enumerate} 
        \item If the energy field is considered to be a Gaussian function with variance parameter $\sigma$, the error, $\ep\left[ \sum_{\tau=1}^m \| \hat\s_{\tau} - \s^\ast\|_2^2 \right]$ is bounded as
        \begin{align}
             \sum_{k,l} C \sigma^2 \log^2m \frac{\|\bm c_{k,l} - \bm c^\ast\|^2}{\exp(-\frac{\|\bm c_{k,l} - \bm c^\ast\|^2}{2n\sigma^2})}
        \end{align}
        \item If the energy field is considered to be a Laplacian function with scale parameter $\gamma$, the average error, $\ep\left[ \sum_{\tau=1}^m \| \hat\s_{\tau} - \s^\ast\|_2^2 \right]$ is bounded by
        \begin{align}
            \sum_{k,l} C \gamma \log^2m \frac{\|\bm c_{k,l} - \bm c^\ast\|^2}{\exp(-\frac{\|\bm c_{k,l} - \bm c^\ast\|}{2n\gamma})}
        \end{align}
        \end{enumerate}
    \end{corollary}
    
    \textbf{Discussion:} Notice from Corollary \ref{cor:special_case} that for both Gaussian and Laplacian fields,  the upper bound is positively correlated with the signal spread ($\sigma^2$ for Gaussians and $\lambda$ for Laplacians).
         The reason is that the maximum-entropy-based sampling ensures that the peak is sampled with high probability for the low spread scenario; this trend is validated in our 
          numerical experiments (see Fig. \ref{fig:varcomp}).  Furthermore, our error upper bounds are higher for the Gaussian energy fields (it is not easy to exactly compare the two settings)\footnote{this follows since the summands scale as $x^2/\exp(-x^2)$ vs $x^2/\exp(-x)$}. This has been observed in previous work \cite{active_target1, umf_sourceloc} as well. We corroborate this trend through experiments in (see Fig \ref{fig:comp_active_passive}). 


    
    

    \begin{figure}[t!]
\centering
\begin{tikzpicture}
 \begin{groupplot}[
     group style={
        group size=2 by 1,
        horizontal sep=2cm,
        vertical sep=1cm,
        x descriptions at=edge bottom,
        },
        enlargelimits=false,
        width = .4\linewidth,
        height=4.5cm,
        enlargelimits=false,
        grid=both,
    grid style={line width=.1pt, draw=gray!10},
    major grid style={line width=.2pt,draw=gray!50},
    minor x tick num=5,
    minor y tick num=5,
     ]
    
    \nextgroupplot[
    legend entries={
	passive ($m = 80$),
	passive ($m=100$),
	\textbf{proposed} ($m=80$),
    \textbf{proposed} ($m=100$),
	},
    legend style={at={(1.1, 1.34)}},
    legend columns = 2,
    legend style={font=\footnotesize},
    xlabel={\small{scale parameter of Laplacian energy, $\gamma$}},
    ylabel={{$\|\hat{\s} - \s^\ast\|_2$}},
    title style={at={(0.5,-.35)},anchor=north,yshift=1},
    xticklabel style= {font=\footnotesize, yshift=-1ex},
    yticklabel style= {font=\footnotesize, xshift=-.5ex},
    ]
        \addplot [teal, line width=1.2pt, mark=o,mark size=3pt] table[x index = {0}, y index = {4}]{\varcomplaplaciannointerp};
     \addplot [olive, line width=1.2pt, mark=square,mark size=3pt] table[x index = {0}, y index = {3}]{\varcomplaplaciannointerp};
    \addplot [red, line width=1.2pt, mark=diamond,mark size=3pt] table[x index = {0}, y index = {1}]{\varcomplaplaciannointerp};
    \addplot [black, line width=1.2pt, mark=Mercedes star,mark size=3pt] table[x index = {0}, y index = {2}]{\varcomplaplaciannointerp};
    
    \nextgroupplot[
	legend entries={
	passive \cite{umf_sourceloc},
	greedy,
	\textbf{proposed},
	},
    legend style={at={(1, 0.6)}},
    legend columns = 1,
    legend style={font=\footnotesize},
	xlabel={\small{noise standard deviation, $\sigma_n$}},
    ylabel={{$\|\hat{\s} - \s^\ast\|_2$}},
    title style={at={(0.4,-.35)},anchor=north,yshift=1},
    xticklabel style= {font=\footnotesize, yshift=-1ex},
    yticklabel style= {font=\footnotesize, xshift=-.5ex},
    ]
     \addplot [black, line width=1.2pt, mark=square,mark size=3pt] table[x index = {0}, y index = {1}]{\noisecompgaussnointerp};
    \addplot [red, line width=1.2pt, mark=o,mark size=3pt] table[x index = {0}, y index = {2}]{\noisecompgaussnointerp};
    \addplot [blue, line width=1.2pt, mark=diamond,mark size=4pt] table[x index = {0}, y index = {3}]{\noisecompgaussnointerp};
    
\end{groupplot}
\end{tikzpicture}
\vspace*{-0.1in}
\caption{\textbf{Left:} Source localization error with respect to the scale parameter of the Laplacian energy field. Source localization error with respect to the noise standard deviation, $\sigma_n$}
\label{fig:varcomp}
\vspace*{-0.1in}
\end{figure}
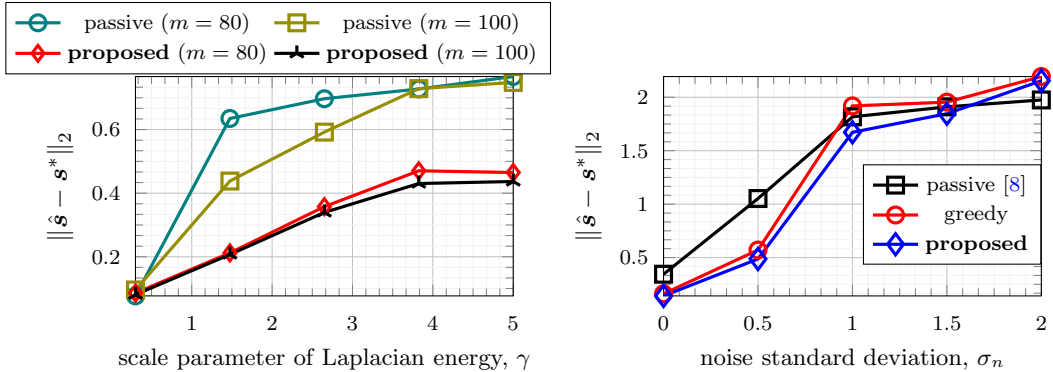

    \section{Numerical Results}
  All experiments are averaged over $100$ independent trials. The codes are available at \url{https://github.com/praneethmurthy/active-localization}.

  \textbf{Effectiveness of Active Sampling Strategy.}
  We first illustrate that the proposed method indeed serves as an effective sampling strategy for the source localization problem. We generate the data as follows: we set the target area, $\mathcal{A} := [-L/2, L/2] \times [-L/2, L/2]$ with $L = 5$. We discretize $\mathcal{A}$ into a $n \times n$ equally spaced grid  with $n = 100$. We consider a single source placed at $\s^\ast := (2,3)^\intercal$. For the energy emitted by the source, in this paper, we consider two cases: (a) Gaussian energy field, i.e., 
      $h(\x, \s^\ast) =  C \exp(-\|\s^\ast - \x\|_2^2/2\sigma^2)$ with variance $\sigma^2$ and (b) Laplacian energy field, i.e., 
      $h(\x, \s^\ast) =  C \exp(-\|\s^\ast - \x\|_1/\lambda)$  with  scale parameter $\lambda$. For the first experiment, we consider the Gaussian setting with $\sigma^2 = 1$, and do not consider additional measurement noise. 
  
  
We compare the proposed method with the state-of-the-art passive source localization approach \cite{umf_sourceloc} and a {\em greedy} baseline that directly exploits the unimodality property of the energy matrix. Concretely, the greedy baseline is essentially Algorithm \ref{algo:active_ucb} without considering the last two terms in Line 7. We implement \cite{umf_sourceloc}  and in each sequential sampling stage, we draw a sample uniformly at random 

We implement \algoname{} as follows. We start with a Latin Squares initialization that queries the energy measurement at $n = 100$ locations. We employ the same method for the greedy baseline as well as the proposed approach. We set the total number of sequential samples, $m = 50$ and plot the localization error of all methods with respect to the number of sequential samples observed. The results are shown in Fig. \ref{fig:comp_active_passive}. Notice that both the Greedy, and the proposed methods perform significantly better than \cite{umf_sourceloc}. Furthermore, the proposed approach performs better than the Greedy approach.  


\textbf{Evaluating Model Parameters.} Next, we illustrate the effect of varying the parameter of the energy field. In particular, we consider the Laplacian setting, and vary the scale parameter. We generate the data exactly as done in the previous experiment. We implement the baseline, \cite{umf_sourceloc}, and \algoname{} with $m = 100$, and illustrate the results in Fig. \ref{fig:varcomp}. We observe that (i) even in the Laplacian setting \algoname{} outperforms the baseline; (ii) as the scale parameter increases, i.e., the steepness of the field reduces, the localization increases. This is in accordance with Corollary \ref{cor:special_case}. We observed similar trends for the Gaussian field as well. 

\textbf{Comparison with respect to Additional noise.} Finally, we investigate the robustness of \algoname{} to additional noise. We generate the data as done in the first experiment, but add independent Gaussian noise with zero-mean and variance $\sigma_n^2$. We implement the passive baseline \cite{umf_sourceloc}, the greedy method, and \algoname{}. We notice that when noise is very high, the passive method is better, but in low-noise regimes, \algoname{} is better. Furthermore, \algoname{} is more robust to noise than the greedy approach as expected.

\section{Conclusions and Future Work}
In this work, we studied the problem of non-parametric active single source localization. We proposed \algoname, an active sampling algorithm that enjoys a error bound that scales as $O(\log^2m/m)$ for $m$ sequential samples. Experimentally, we showed that the proposed method performs well across several different types of data and is robust to noise. As part of future work, we will consider (i) dealing with multiple sources; (ii) development of complete convergence guarantees for the proposed method; and (iii) development of a more efficient algorithm that does not involve explicit matrix completion.

  \appendices
  \renewcommand\thetheorem{\arabic{section}.\arabic{theorem}}
  \counterwithin{theorem}{section}

  \section{Proof Sketch: Details}
  In this section, we provide the proof of the main result, Theorem \ref{thm:main_result}. To this end, we first prove 
  a generalized version of the initialization result, 
  Lemma \ref{lem:init_lemma}, followed by the sequential sampling result. 
 
    
    \begin{lemma}[Latin Squares Initialization]\label{lem:init_lemma_all}
    Consider the true energy matrix, $\Y$. With probability at least $1 - n^{-10}$ 
    \begin{enumerate}
        \item the error bound for the matrix is
    \begin{align}
        \|\Y - \Y_{\Omega_{\init}}\|_2 \leq 1.01 (1-1/n) \|\u\|_1 \|\v\|_2
    \end{align}
    
    \item the error bound for the subspaces is 
    \begin{align}
        ({\u}^\intercal \tilde{\u}_0)^2 + ({\v}^\intercal\tilde{\v}_0)^2 \leq \frac{2 \|\Y - \Y_{\Omega}\|_2^2}{(\|\Y\| - \|\Y - \Y_{\Omega}\|)^2}
    \end{align}
    where $\tilde{\u}_0$ and $\tilde{\v}_0$  are the top-$1$ left, and right singular vectors of $\Y_{\Omega}$ repsectively. 
    \end{enumerate}
    \end{lemma}
    

\newcommand{\baru}{\bar{\u}}
\newcommand{\barv}{\bar{\v}}

    \begin{proof}[Proof of Lemma \ref{lem:init_lemma_all}]

    \textbf{Proof of Item 1:} Define $\baru := \sigma_y \u$  and $\barv := \sigma_y \v$.  For simplicity, we drop the subscript $\init$ in this proof. Notice that
    \begin{align*}
        \|\Y - \Y_{\Omega}\|_2 = \|\baru\v^\intercal - (\baru\v^\intercal)_{\Omega} \|_2 
        = \norm{\sum_{i,j} \phi_{i,j} \e_i \baru_i \barv_j \e_j^\intercal }
    \end{align*}
    where, $\e_i$ is the $i$-th canonical basis vector in $\R^n$ and 
    we use $\phi_{i,j} 1$ if $(i,j) \in \Omega^c$ and $0$ otherwise. 
    Although these matrices are not independent, notice that for a fixed $i$, the $\phi_{i,j}$ are independent Bernoulli r.v.'s. For a fixed $i$, we can apply Matrix Bernstein \cite{tropp}. 
    Let $\X_{i,j} := \phi_{i,j} \e_i \baru_i \barv_j \e_j^\intercal$. We have 
    \begin{align*}
        \|\ep[\X_{i,j}] \| &= \ep[\phi_{i,j}]\|\e_i \baru_i \barv_j \e_j\| = (1-1/n) \cdot \baru_i \barv_j \\
        &\implies \max_{j} \|\ep[\X_{i,j}]\| \leq (1-1/n) \baru_i \barv_{j^\ast} := R_i
    \end{align*}
    
    We can also bound the  ``variance parameter''
    \begin{align}\label{eq:sig_i_one}
        \norm{\sum_j \ep[\X_{i,j} \X_{i,j}^\intercal]} 
        &= \left(1 - \frac{1}{n} \right) \baru_i^2 \|\barv\|^2
    \end{align}
where the last line follows from the observation that the matrix is a matrix with just one non-zero element at the $(i,i)$-location. Furthermore, we have 
    \begin{align}\label{eq:sig_i_two}
        \norm{\sum_j \ep[\X_{i,j}^\intercal \X_{i,j}]} 
        &= \left(1 - \frac{1}{n} \right) \baru_i^2 \barv_{j^\ast}^2
    \end{align}
where the last line follows from the observation that the spectral norm of a diagonal matrix is the value of the largest element, and the unimodality property of $\barv$. From these, we get, 
\begin{align*}
\max\left\{ \norm{\sum_j \ep[\X_{i,j} \X_{i,j}^\intercal]},   \norm{\sum_j \ep[\X_{i,j}^\intercal \X_{i,j}]} \right\} = \left(1 - \frac{1}{n} \right) \baru_i^2 \|\barv\|^2 := \sigma_i^2
\end{align*}
Applying the matrix Bernstein for a fixed $i$, we have 
\begin{align*}
          \Pr\left( \norm{\sum_{j} \X_{i,j}} \geq \norm{\sum_j \ep[\X_{i,j}]} + s \right) 
          \leq 2n \exp\left( \frac{-s^2/2}{\sigma_i^2 + R_is/3} \right)
      \end{align*}
    now, picking $s := \sqrt{\frac{0.01 \cdot 12 \log n}{2 (1-1/n) \baru_i}} (1-1/n) \baru_i \|\barv\|$ gives us that 
    \begin{align*}
        \Pr\left( \norm{\sum_j \phi_{i,j} \e_i \baru_i \barv_j \e_j^\intercal} \geq 1.01 (1-1/n) \baru_i \|\barv\| \right) \leq n^{-11}
    \end{align*}
    A union bound over all $i$ concludes the proof. 
    
    \textbf{Proof of Item 2:}
    The proof follows form the result of Item $1$ followed by an application of Wedin's Theorem \cite{wedin}.
    Define 
    \begin{align*}
        \Y_{\Omega} &\overset{SVD}{=} \sum_{i=1}^n \tilde{\sigma_{0,i}} \tilde{\u}_{0,i}  \tilde{\u}_{0,i}^\intercal = \sigma_{0,1} \tilde{\u}_{0,1} \tilde{\v}_{0,1}^\intercal +  \sum_{i=2}^n \tilde{\sigma_{0,i}} \tilde{\u}_{0,i}  \tilde{\u}_{0,i}^\intercal := \tilde{\u}_0 \tilde{\sigma}_0 \tilde{\v}_0^\intercal + \tilde{\U}_{0, \perp} \tilde{\bm{\Sigma}}_{0,\perp} \tilde{\V}_{0,\perp}^\intercal
    \end{align*}
    observe that the deviation parameter, $\delta := \min\{ \min_{1 \leq i \leq r, r \leq j \leq n_2} |\sigma_i - \tilde{\sigma}_j|, \min_{1 \leq i \leq r} \sigma_i \} = \min\{ \sigma_y^2 - \|\bm{\Sigma}_{0,\perp}\|, \sigma_y^2\} = \sigma_y^2 - \|\bm{\Sigma}_{0,\perp}\|_2 > 0$. Thus, 
    \begin{align}
        (\u^\intercal \tilde{\u}_0)^2 + ({\v}^\intercal\tilde{\v}_0)^2 \leq \frac{\|{\u}^\intercal(\Y - \Y_{\Omega})\|_2^2 + \|{\v}^\intercal(\Y - \Y_{\Omega})\|_2^2 }{\delta^2}
    \end{align}
    using triangle inequality on the numerator terms and, Weyl's inequality completes the proof. 
    \end{proof}



    \begin{proof}[Proof of Theorem \ref{thm:main_result}]
        First, for the sake of simplicity, assume that the left singular value estimates have the same distribution at all time $t$, i.e,
        $\hat{\u}^t \overset{i.i.d}{\sim} \mathcal{N}(\bm\mu_u, \bm\Sigma_{u})$ with  $\bm\Sigma_u$ diagonal. 
        
        Define $T_{k}^u(m)$ denote the number of times the $k$-th row is selected after $m$ sequential samples. 
        \cite[Theorem 2]{mab_unknown_var} tells us that the (expected) number of times the ``wrong row'' 
        is selected can be bounded as 
        \begin{align}
            \ep[\hat T_k^u(m)] \leq C \left( \frac{(\bm\Sigma_u)_{k,k} + 2b (\Delta^u_k)}{(\Delta^u_k)^2} \right) \log m 
        \end{align}
        where $b$ is the bound on the energy field and $\Delta_k^u := \bm\mu_u^\ast - (\bm\mu_u)_{k}$ is the sub-optimality gap of the $k$-th row.
        Similarly, assuming that the right singular value estimates have the same distribution at all time $t$, i.e,
        $\hat{\v}^t \overset{i.i.d}{\sim} \mathcal{N}(\bm\mu_v, \bm\Sigma_{v})$ with  $\bm\Sigma_v$ diagonal, the number of the $l$-th colum is selected after $m$ sequential samples can be bounded as 
        \begin{align}
            \ep[\hat T_l^v(m)] \leq C \left( \frac{(\bm\Sigma_v)_{l,l} + 2b \Delta^v_l}{(\Delta^v_l)^2} \right) \log m 
        \end{align}
        where $\Delta_l^v := \bm\mu_v^\ast - (\bm\mu_v)_{l}$ is the sub-optimality gap of the $l$-th column. Now consider the expected error function, 
        \begin{align*}
            \ep\left[\sum_{\tau=1}^m \|\hat{\s}_{\tau} - \s^\ast\|\right] &= \ep\left[\sum_{\tau = 1}^m \sum_{k,l} \mathds{1}_{\{i_\tau = k, j_{\tau} =l\}} (|k - i^\ast|^2 + |l - j^\ast|^2) \right] \\
            &\overset{(a)}{=} \sum_{k,l} \ep[T_k^u(m)]\ \ep[T_l^v(m)] (|k-i^\ast|^2 + |l-j^\ast|^2) \\
        \end{align*}
        where (a) follows from the definition of $T_k^u(m), T_l^v(m)$, re-ordering the summation, and using the near-independence result  \cite{uq_matcomp} to decompose the expected value completes the proof. 
    \end{proof}

\bibliographystyle{IEEEtran}
\bibliography{active_loc}    

\end{document}